\documentclass[11pt,a4paper,reqno]{amsart}
\usepackage{multirow}
\usepackage[centertags]{amsmath}
\usepackage{amsfonts}
\usepackage{amssymb}
\usepackage{amsthm}
\usepackage{newlfont}
\usepackage{a4}
\usepackage{enumerate}
\usepackage[pdftex]{graphicx,color}
\theoremstyle{definition}
\newtheorem{defn}{Definition}[section]
\newtheorem{thm}[defn]{Theorem}

\newtheorem{tvr}[defn]{Proposition}

\theoremstyle{remark}
\newtheorem{example}{Example}[section]

\usepackage{bbm}

\newcommand{\ii}{\mathrm{i}}
\newcommand{\map}{\rightarrow}

\renewcommand{\int}{\operatorname{int}}

\newcommand{\q}{\quad}

\renewcommand{\epsilon}{\varepsilon}

\newcommand{\ep}{\varepsilon}
\newcommand{\la}{\lambda}
\newcommand{\al}{\alpha}
\newcommand{\om}{\omega}
\renewcommand{\rho}{\varrho}
\renewcommand{\phi}{\varphi}

\newcommand{\R}{{\mathbb{R}}}

\newcommand{\N}{{\mathbb N}}

\newcommand{\Com}{{\mathbb C}}

\newcommand{\Z}{\mathbb{Z}}

\newcommand{\set}[2]{\left\{#1 \, |\, #2 \right\}}
\newcommand{\setb}[2]{\left\{#1 \, \mid\, #2 \right\}}

\newcommand{\abs}[1]{\left\vert#1\right\vert}
\newcommand{\wt}{\widetilde}

\newcommand{\sca}[2]{\langle #1,\, #2\rangle}

\begin{document}

\title[Torus $E-$discretization]
{On $E-$Discretization of Tori of Compact Simple Lie Groups}

\author{Ji\v{r}\'{\i} Hrivn\'{a}k$^{1,2}$}
\author{Ji\v{r}\'{\i} Patera$^1$}

\date{\today}
\begin{abstract}\
Three types of numerical data are provided for compact simple Lie groups $G$ of classical types and of any rank. This data is indispensable for Fourier-like expansions of multidimensional digital data into finite series of $E-$functions on the fundamental domain $F^{e}$. Firstly, we determine the number $|F^{e}_M|$ of points in $F^{e}$ from the lattice $P^{\vee}_M$, which is the refinement of the dual weight lattice $P^{\vee}$ of $G$ by a positive integer $M$. Secondly, we find the lowest set $\Lambda^{e}_M$ of the weights, specifying the maximal set of $E-$functions that are pairwise orthogonal on the point set $F^{e}_M$. Finally, we describe an efficient algorithm for finding the number of conjugate points to every point of $F^{e}_M$. Discrete $E-$transform, together with its continuous interpolation, is presented in full generality.

\end{abstract}\

\maketitle

\noindent
$^1$ Centre de recherches math\'ematiques,
         Universit\'e de Montr\'eal,
         C.~P.~6128 -- Centre ville,
         Montr\'eal, H3C\,3J7, Qu\'ebec, Canada; patera@crm.umontreal.ca\\
$^2$ Department of physics,
Faculty of nuclear sciences and physical engineering, Czech
Technical University, B\v{r}ehov\'a~7, 115 19 Prague 1, Czech
republic; jiri.hrivnak@fjfi.cvut.cz


\section{Introduction}

The $E-$discretization of this paper differs in an important way, both theoretically and practically, from the `ordinary' discretization studied in \cite{HP}. First, we point out what the two approaches have in common, then we  underline their differences.

Both discretizations must have an underlying compact simple or semisimple Lie group $G$ of any rank $n<\infty$ and of any type. The rank is equal to the number of variables involved in the process. Then we introduce new classes of multivariate special functions, the $C-$ and $S-$functions in \cite{HP}, and the $E-$functions (\ref{E}) here. The functions are orthogonal on a finite region $F$ of the $n-$dimensional real Euclidean space $\R^n$, as continuous functions as well as functions restricted to a fragment of a lattice $L\cap F\subset\R^n$. The lattice can have any density but its symmetry is imposed by the underlying Lie group. Thus in $\R^n$ there are as many different lattices and special functions orthogonal on them as there are semisimple Lie groups of rank $n$. The families of $C-$, $S-$, and $E-$functions were recognized and named in \cite{P}. They generalize the common cosine, sine, and exponential functions of one variable. Many of their properties are described in reviews \cite{KP1,KP2,KP3}.

The orbit function $C$ and $S$ are built using the finite reflection group $W$, attached to each $G$, and called the Weyl group of $G$. The $E-$functions are built using the even subgroup $W^e\subset W$ which is not a reflection group. Much less specific information is available about such groups in the literature. The $E-$functions are simpler, for the same $G$, than the orbit functions of type $C$ or $S$. More precisely, an orbit function is a sum/difference of two $E-$functions. They have no prescribed behavior at the boundary $\partial F^e$ of the region of their orthogonality $F^e$, unlike the orbit functions which are either symmetric or antisymmetric with respect to their boundary $\partial F$. Region $F^e$ is a union of two adjacent copies of region $F$. Discretizations of the functions over lattice fragments $L\cap F$ and $L\cap F^e$, particularly their orthogonality over  $L\cap F^e$, require different specific values of a number of constants required for groups $G$.

The purpose of the paper is to provide all of the information needed for the Fourier analysis of $n-$dimensional digital data in terms of their $E-$function expansions, in the context of an admissible symmetry group $G$. We suppose that $G$ is one of the classical simple Lie groups with the Lie algebras of types $A_n$, $B_n$, $C_n$, and $D_n$.

In Section~2, some standard properties of simple Lie groups and/or simple Lie algebras are recalled. The properties of group $W^e$ not generally available elsewhere are important. In Section~3, we describe the lattice grids $F^e_M\subset\R^n$, where the digital data is provided. The density of the grid is controlled by our choice of the integer $M$. For any grid $F^e_M$, there are only finitely many distinct $E-$functions that are orthogonal on the grid. Also, the lowest set of such functions is described. The functions are labeled by the grid of points $\Lambda^{e}_M$. In Section~4, the properties of the $E-$functions are described for each point of $\Lambda^{e}_M$. Discrete $E-$transforms are presented in Section~5. Concluding comments and remarks are found in Section~6.

\section{Pertinent properties of simple Lie groups and their Lie algebras}

\subsection{Definitions and notations}\

Consider the Lie algebra of the compact simple Lie group $G$ of rank $n$, with the set of simple roots $\Delta=\{\al_1,\dots,\al_n\}$, spanning the Euclidean space $\R^n$ \cite{Bour,H1,VO}.

By uniform and standard methods for $G$ of any type and rank, a number of related quantities and virtually all the properties of $G$ are determined from $\Delta$. We make use of the following:

\begin{itemize}
\item
The highest root $\xi\equiv -\al_0=m_1\al_1+\dots+m_n\al_n$. Here the coefficients $m_j$ are known positive integers, also called the marks of $G$.

\item
The Coxeter number $m=1+m_1+\dots+m_n$ of $G$.

\item
The Cartan matrix $C$
\begin{equation*}
 C_{ij}=\frac{2\sca{\al_i}{\al_j} }{\sca{\al_j}{\al_j}},\q i,j\in\{1,\dots,n\}.
\end{equation*}

\item
The order $c$ of the center of $G$,
\begin{equation}\label{Center}
 c=\det C.
\end{equation}
\item
The dual weight lattice,
\begin{equation*}
 P^{\vee}=\set{\om^{\vee}\in \R^n}{\sca{\om^{\vee}}{\al}\in\Z,\, \forall \al \in \Delta}=\Z \om_1^{\vee}+\dots +\Z \om_n^{\vee}\,.
 \end{equation*}
\item
The root lattice $Q$ of $G$,
 \begin{equation}\label{dPPP}
 Q=\set{\al\in \R^n}{\sca{\al}{\om^{\vee}}\in\Z,\,\forall\om^{\vee}\in P^{\vee}}=\Z\al_1+\dots+\Z\al_n\,.
\end{equation}

\item
The dual root lattice
\begin{equation*}
 Q^{\vee}=\Z \al_1^{\vee}+\dots +\Z \al^{\vee}_n\,,\quad\text{where}\quad
     \al^{\vee}_i=\frac{2\al_i}{\sca{\al_i}{\al_i}}\,.
\end{equation*}
\end{itemize}

\subsection{Weyl group and affine Weyl group}\

The properties of Weyl groups and affine Weyl groups can be found for example in \cite{H2,BB}.
The finite Weyl group $W$ is generated by $n$ reflections $r_\al$, $\al\in\Delta$, in $(n-1)$-dimensional `mirrors' orthogonal to simple roots intersecting at the origin:
\begin{equation*}
r_{\al_i}a\equiv r_i a=a-\frac{2\sca{a}{\al_i} }{\sca{\al_i}{\al_i}}\al_i\,,
\qquad a\in\R^n\,.
\end{equation*}

The infinite affine Weyl group $W^{\mathrm{aff}}$ is the semidirect product of the Abelian group of translations $Q^\vee$ and of the Weyl group~$W$.
\begin{equation}\label{direct}
 W^{\mathrm{aff}}= Q^\vee \rtimes W.
\end{equation}
Equivalently, $W^{\mathrm{aff}}$ is generated by reflections $r_i$ and reflection $r_0$, where
\begin{equation*}
r_0 a=r_\xi a + \frac{2\xi}{\sca{\xi}{\xi}}\,,\qquad
r_{\xi}a=a-\frac{2\sca{a}{\xi} }{\sca{\xi}{\xi}}\xi\,,\qquad a\in\R^n\,.
\end{equation*}

The fundamental region $F$ of $W^{\mathrm{aff}}$ is the convex hull of the points $\left\{ 0, \frac{\om^{\vee}_1}{m_1},\dots,\frac{\om^{\vee}_n}{m_n} \right\}$:
\begin{align}
  F &=\setb{y_1\om^{\vee}_1+\dots+y_n\om^{\vee}_n}{y_0,\dots, y_n\in \R^{\geq 0}, \, y_0+y_1 m_1+\dots+y_n m_n=1  } \label{deffun} \\
&= \setb{a\in \R^n}{\sca{a}{\al}\geq 0,\forall \al\in\Delta,\sca{a}{\xi}\leq 1  }\nonumber
\end{align}

Since $F$ is a fundamental region of $W^{\mathrm{aff}}$, we have:
\begin{enumerate}
\item For any $a\in \R^n$ there exists $a'\in F$, $w\in W$ and $q^{\vee}\in Q^{\vee}$ such that
\begin{equation}\label{fun1}
 a=wa'+q^{\vee}.
\end{equation}

\item If $a,a'\in F$ and $a'=w^{\mathrm{aff}}a$, $w^{\mathrm{aff}}\in W^{\mathrm{aff}}$ then $a=a'$, i.e. if there exist $w\in W$ and $q^{\vee}\in Q^{\vee}$ such that $a'=wa+q^{\vee}$, then
\begin{equation}\label{fun2}
 a'=a=wa+q^{\vee}.
\end{equation}

\item
Consider a point
$a=y_1\om^{\vee}_1+\dots+y_n\om^{\vee}_n\in F$, such that
$y_0+y_1 m_1+\dots+y_n m_n=1$.
The isotropy group
 \begin{equation}\label{stab}
 \mathrm{Stab}_{W^{\mathrm{aff}}}(a) = \setb{w^{\mathrm{aff}}\in W^{\mathrm{aff}}}{w^{\mathrm{aff}}a=a}
\end{equation}
of the point $a$ is trivial, $\mathrm{Stab}_{W^{\mathrm{aff}}}(a)=1$, if $a\in \mathrm{int}(F)$, where $\mathrm{int}(F)$ denotes the interior of $F$, i.e. all $y_i>0$, $i=0,\dots,n$. Otherwise the group $\mathrm{Stab}_{W^{\mathrm{aff}}}(a)$ is finite and generated by such $r_i$ for which $y_i=0$, $i=0,\dots,n$.
\end{enumerate}

\subsection{Even Weyl group and even affine Weyl group}\

Elements of the Weyl group $W$ are orthogonal linear transformations of the space $\R^n$. A subgroup of $W$ of the elements $w\in W$ with determinant $\det w=1$ is called the even Weyl group $W^e$. i.e.
\begin{equation*}
 W^e = \set{w\in W}{\det w=1}
\end{equation*}

The subgroup $W^e$ can be viewed as the kernel of the homomorphism $\det : W \ni w \mapsto \det w$.
Since $\ker \det = W^e$ and $\det (W)= \{\pm 1\}$, group $W^e$ is a normal subgroup of $W$ and
\begin{equation*}
 \abs{W}= 2\abs{W^e}.
\end{equation*}

The infinite even affine Weyl group $W^{\mathrm{aff}}_e$ is the semidirect product of the group of translations $Q^\vee$ and of the even Weyl group~$W^e$
\begin{equation}\label{directE}
 W^{\mathrm{aff}}_e= Q^\vee \rtimes W^e.
\end{equation}

We choose some fixed $j\in \{1,\dots,n\} $ and define the set $F^e$ by
\begin{equation}\label{fundE}
 F^e = F\cup r_j \int (F).
\end{equation}
Note that $F^e$ consists of two disjoint parts: the closed simplex $F$ and the open interior of the simplex $r_j \int (F)$. From this decomposition, we also obtain the formula for the volume of $F^e$:
$$\mathrm{vol} (F^e)=2\,\mathrm{vol}(F).$$
In the following proposition, we show that $F^e$ is a fundamental region of the even affine Weyl group~$W^{\mathrm{aff}}_e$.

\begin{tvr}\label{tvrfund}
For the set $F^e$, it holds that
\begin{enumerate}
\item For any $a\in \R^n$ there exists $a'\in F^e$, $w\in W^e$ and $q^{\vee}\in Q^{\vee}$, such that
\begin{equation}\label{fun1E}
 a=wa'+q^{\vee}.
\end{equation}

\item If $a,a'\in F^e$ and $a'=w^{\mathrm{aff}}a$, $w^{\mathrm{aff}}\in W^{\mathrm{aff}}_e$ then $a=a'$, i.e. if there exist $w\in W^e$ and $q^{\vee}\in Q^{\vee}$ such that $a'=wa+q^{\vee}$, then
\begin{equation}\label{fun2E}
 a'=a=wa+q^{\vee}.
\end{equation}

\item
Consider a point $a\in F^e$. If $a\in \int(F)$ or $a\in r_j \int(F)$, then the
isotropy group
 \begin{equation}\label{stabE1}
 \mathrm{Stab}_{W^{\mathrm{aff}}_e}(a) = \setb{w^{\mathrm{aff}}\in W^{\mathrm{aff}}_e}{w^{\mathrm{aff}}a=a}
\end{equation}
is trivial, $\mathrm{Stab}_{W^{\mathrm{aff}}_e}(a)=1$. If $a\in F \setminus \int(F) $, then it holds that
 \begin{equation}\label{stabE2}
 \abs{\mathrm{Stab}_{W^{\mathrm{aff}}}(a)} = 2\abs{\mathrm{Stab}_{W^{\mathrm{aff}}_e}(a)}.
\end{equation}
\end{enumerate}
\end{tvr}
\begin{proof}
\begin{enumerate}

\item Suppose we have some $a\in \R^n$. It follows from (\ref{fun1}) that there exist $a'\in F$, $w\in W$ and $q^{\vee}\in Q^{\vee}$ such that
$a=wa'+q^{\vee}$. If $\det w =1$ then we have found $a'\in F\subset F^e$, $w\in W^e$ and $q^{\vee}\in Q^{\vee}$ that satisfy (\ref{fun1}). Suppose that $\det w =-1$ and
\begin{enumerate}
\item $a'\in \int(F)$. Since $F^e \equiv F\cup r_j \int (F)$ we have $r_j a' \in F^e$. Taking into account that $\det r_j=-1$ and $r_j^2=1$ we obtain $wr_j \in W^e $ and $a=(wr_j) r_ja'+q^{\vee}$.
\item $a'\in F\setminus\int(F)$. We have from (\ref{stab}) that the stabilizer $\mathrm{Stab}_{W^{\mathrm{aff}}}(a')$ is non-trivial and contains some $r_i$, $i\in\{0,\dots,n\}$ such that $r_ia'=a'$. If $i\in\{1,\dots,n\}$ then we have $wr_i\in W^e$ and $a=(wr_i)a'+q^\vee$. If $i=0$ then we have $wr_\xi\in W^e$ and $a=(wr_\xi)a'+q'^\vee$ where $q'^\vee=2w\xi/\sca{\xi}{\xi}+q^\vee\in Q^\vee$.
\end{enumerate}

\item Suppose we have $a,a'\in F^e $ and $w\in W^e$, $q\in Q^\vee$ such that
\begin{equation}\label{funproof}
 a'=wa+q^{\vee}.
\end{equation}
Since $F^e$ consists of two disjoint parts $F$ and $r_j \int (F)$, we distinguish the following cases:
\begin{enumerate}
\item $a,a'\in F$. It follows immediately from (\ref{fun2}) that $a=a'$.
\item $a,a'\in r_j\int (F)$. Consider $b,b'\in \int(F)$ such that $a=r_j b$ and $a'=r_jb'$. Then $b'=r_jwr_j b+r_jq^\vee$ and from (\ref{fun2}) we obtain $b=b'$, i.e. $a=a'$.
\item $a'\in F$, $a\in r_j\int (F)$. Consider $a=r_jb$, $b\in\int (F)$. Then $a'=wr_jb+q^\vee $ and from (\ref{fun2}) we have that $a'=b$. Since the stabilizer of the point $b\in\int (F)$ is trivial, $\mathrm{Stab}_{W^{\mathrm{aff}}}(b)=1$, we obtain $wr_j=1$. We conclude that $\det w=-1$, which is contradictory to the assumption $w\in W^e$ in (\ref{funproof}) and thus, this case cannot occur.
\end{enumerate}
\item If $a\in \int(F)$, then from (\ref{stab}) we have that the stabilizer $\mathrm{Stab}_{W^{\mathrm{aff}}}(a)$ is trivial. Since the stabilizers of the points $a$ and $r_ja$ are conjugated, the stabilizer $\mathrm{Stab}_{W^{\mathrm{aff}}}(r_ja)$ is also trivial. Then since $\mathrm{Stab}_{W^{\mathrm{aff}}_e}(a)\subset \mathrm{Stab}_{W^{\mathrm{aff}}}(a)$, we have $\mathrm{Stab}_{W^{\mathrm{aff}}_e}(a)=\mathrm{Stab}_{W^{\mathrm{aff}}_e}(r_ja)=1$.

We have from (\ref{direct}) that for any $w^{\mathrm{aff}}\in W^{\mathrm{aff}}$ there exist a unique $w\in W$ and a unique shift $T(q^{\vee})$ such that $w^{\mathrm{aff}}=T(q^{\vee})w$. Define a homomorphism $\tau: \mathrm{Stab}_{W^{\mathrm{aff}}}(a)\map \{\pm 1\} $ for $w^{\mathrm{aff}}\in \mathrm{Stab}_{W^{\mathrm{aff}}}(a)$ by $\tau(w^{\mathrm{aff}})=\tau(T(q^{\vee})w)=\det w$.
If $a\in F \setminus \int(F)$  then we have from (\ref{stab}) that the stabilizer $\mathrm{Stab}_{W^{\mathrm{aff}}}(a)$ is non-trivial, finite and contains some $r_i$, $i\in\{0,\dots,n\}$. Since $\tau (r_i)=-1$, $i\in\{0,\dots,n\}$ holds, we have $\tau(\mathrm{Stab}_{W^{\mathrm{aff}}}(a))= \{\pm 1\}$. Then since $\ker \tau =\mathrm{Stab}_{W^{\mathrm{aff}}_e}(a) $, we conclude that $\mathrm{Stab}_{W^{\mathrm{aff}}}(a)/\mathrm{Stab}_{W^{\mathrm{aff}}_e}(a)\cong \{\pm 1\} $.
\end{enumerate}
\end{proof}

\subsection{Action of $W^e$ on the maximal torus $\R^n/Q^{\vee}$}\

If we have two elements $a,a'\in \R^n$ such that $a'-a=q^{\vee}$, with $q^{\vee}\in Q^{\vee}$, then for $w\in W^e$ we have $wa-wa'=wq^{\vee}\in Q^{\vee}$, i.e. we have a natural action of $W^e$ on the torus $\R^n/Q^{\vee}$. For $x\in \R^n/Q^{\vee}$ we denote the isotropy group and its order by
 \begin{equation}\label{hxe}
h^e_x\equiv |\mathrm{Stab^e} (x)|,\q\mathrm{Stab^e} (x)=\set{w\in W^e}{wx=x}
\end{equation}
We denote the orbit and its order by
 \begin{equation*}
\ep^e(x)\equiv |W^ex|,\q W^e x=\set{wx\in \R^n/Q^{\vee} }{w\in W^e}.
\end{equation*}
Clearly we have
\begin{equation}\label{epE}
\ep^e(x)=\frac{|W^e|}{h^e_x}.
\end{equation}

\begin{tvr}\begin{enumerate}
\item For any $x\in \R^n/Q^{\vee}$ there exists $x'\in F^e \cap \R^n/Q^{\vee} $ and $w\in W^e$ such that
\begin{equation}\label{rfun1E}
 x=wx'.
\end{equation}

\item If $x,x'\in F^e \cap \R^n/Q^{\vee} $ and $x'=wx$, $w\in W^e$ then
\begin{equation}\label{rfun2E}
 x'=x=wx.
\end{equation}
\item If $x\in F^e \cap \R^n/Q^{\vee} $, i.e. $x=a+Q^{\vee}$, $a\in F^e$ then
\begin{equation}\label{rfunstabE}
\mathrm{Stab^e} (x) \cong \mathrm{Stab}_{W^{\mathrm{aff}}_e}(a).
\end{equation}
\end{enumerate}
\end{tvr}
\begin{proof}
\begin{enumerate}
 \item Follows directly from (\ref{fun1E}).
\item Follows directly from (\ref{fun2E}).
\item We have from (\ref{directE}) that for any $w^{\mathrm{aff}}\in W^{\mathrm{aff}}_e$ there exist a unique $w\in W^e$ and a unique shift $T(q^{\vee})$ such that $w^{\mathrm{aff}}=T(q^{\vee})w$. Define a homomorphism $\psi: \mathrm{Stab}_{W^{\mathrm{aff}}_e}(a)\map W^e $ for $w^{\mathrm{aff}}\in \mathrm{Stab}_{W^{\mathrm{aff}}_e}(a)$ by $\psi(w^{\mathrm{aff}})=\psi(T(q^{\vee})w)=w$.
If $a=w^{\mathrm{aff}}a=wa+q^{\vee}$ then $a-wa=q^{\vee}\in Q^{\vee}$, i.e. we obtain $w\in \mathrm{Stab^e} (x)$ and vice versa. Thus, $\psi (\mathrm{Stab}_{W^{\mathrm{aff}}_e}(a))=\mathrm{Stab^e}(x)$ holds. We also have
\begin{equation*}
\mathrm{ker}\,\psi= \left\{ T(q^{\vee})\in \mathrm{Stab}_{W^{\mathrm{aff}}_e}(a) \right\}=1.
\end{equation*}
\end{enumerate}

\end{proof}

\subsection{Dual Lie algebra}\

The set of simple dual roots $\Delta^{\vee}=\{\al^{\vee}_1,\dots , \al^{\vee}_n\}$ is a system of simple roots of some simple Lie algebra. The system
$\Delta^{\vee}$ also spans Euclidean space $\R^n$.

The dual system $\Delta^{\vee}$ determines:
\begin{itemize}
\item
The highest dual root $\eta\equiv -\al_0^{\vee}= m_1^{\vee}\al_1^{\vee} + \dots + m_n^{\vee} \al_n^{\vee}$. Here the coefficients $m^{\vee}_j$ are called the dual marks of $G$.

\item
The dual Cartan matrix $C^{\vee}$
 \begin{equation*}
 C^{\vee}_{ij}=\frac{2\sca{\al^{\vee}_i}{\al^{\vee}_j} }{\sca{\al^{\vee}_j}{\al^{\vee}_j}}=C_{ji},\q i,j\in\{1,\dots,n\}.
\end{equation*}

\item
The dual root lattice
 \begin{equation*}
 Q^{\vee}=\Z \al^{\vee}_1+\dots +\Z \al^{\vee}_n.
\end{equation*}

\item
The root lattice
\begin{equation*}
 Q=\Z \al_1+\dots +\Z \al_n\,,\quad\text{where}\quad \al_i=\frac{2\al_i^{\vee}}{\sca{\al_i^{\vee}}{\al_i^{\vee}}}.
\end{equation*}

\item The $\Z$-dual lattice
\begin{equation*}
 P=\set{\om\in \R^n}{\sca{\om}{\al^{\vee}}\in\Z,\, \forall \al^{\vee} \in \Delta^{\vee}}=\Z \om_1+\dots +\Z \om_n.
\end{equation*}

\end{itemize}
\subsection{Dual affine Weyl group and its even subgroup}\

Dual affine Weyl group $\widehat{W}^{\mathrm{aff}}$ is a semidirect product of the group of shifts $Q$ and the Weyl group $W$
\begin{equation}\label{directd}
 \widehat{W}^{\mathrm{aff}}= Q \rtimes W.
\end{equation}
Equivalently, $\widehat{W}^{\mathrm{aff}}$ is generated by reflections $r_i$ and reflection $r_0^{\vee}$, where
\begin{equation*}
r_0^{\vee} a=r_{\eta} a + \frac{2\eta}{\sca{\eta}{\eta}}, \q r_{\eta}a=a-\frac{2\sca{a}{\eta} }{\sca{\eta}{\eta}}\eta,\q a\in\R^n.
\end{equation*}

The fundamental region $F^\vee$ of $\widehat{W}^{\mathrm{aff}}$ is the convex hull of the vertices $\left\{ 0, \frac{\om_1}{m^{\vee}_1},\dots,\frac{\om_n}{m^{\vee}_n} \right\}$:

\begin{align}
  F^\vee &=\setb{z_1\om_1+\dots+z_n\om_n}{z_0,\dots, z_n\in \R^{\geq 0}, \, z_0+z_1 m_1^{\vee}+\dots+z_n m^{\vee}_n=1  }\label{defdfun} \\
&= \setb{a\in \R^n}{\sca{a}{\al^{\vee}}\geq 0,\forall \al^{\vee}\in\Delta^{\vee},\sca{a}{\eta}\leq 1  }\nonumber
\end{align}

The dual even affine Weyl group $\widehat{W}^{\mathrm{aff}}_e$ is the semidirect product of the group of translations $Q$, and of the even Weyl group~$W^e$
\begin{equation}\label{directEd}
 \widehat{W}^{\mathrm{aff}}_e= Q \rtimes W^e.
\end{equation}

We choose some fixed $j\in \{1,\dots,n\} $ and define the set $F^{e\vee}$ by
\begin{equation}\label{fundEd}
 F^{e\vee} = F^\vee \cup r_j \int (F^\vee).
\end{equation}

Analogously to Proposition \ref{tvrfund}, we obtain that $F^{e\vee}$ is a fundamental region of the dual even affine Weyl group~$\widehat{W}^{\mathrm{aff}}_e$.
\begin{tvr}
For the set $F^{e\vee}$, it holds that
\begin{enumerate}
\item For any $a\in \R^n$, there exists $a'\in F^{e\vee}$, $w\in W^e$ and $q\in Q$ such that
\begin{equation}\label{fun1Ed}
 a=wa'+q.
\end{equation}

\item If $a,a'\in F^{e\vee}$ and $a'=w^{\mathrm{aff}}a$, $w^{\mathrm{aff}}\in \widehat{W}^{\mathrm{aff}}_e$ then $a=a'$, i.e. if there exist $w\in W^e$ and $q\in Q$ such that $a'=wa+q$ then
\begin{equation}\label{fun2Ed}
 a'=a=wa+q.
\end{equation}

\item
Consider a point $a\in F^{e\vee}$. If $a\in \int(F^{\vee})$ or $a\in r_j \int(F^{\vee})$, then the
isotropy group
 \begin{equation}\label{stabEd}
 \mathrm{Stab}_{\widehat{W}^{\mathrm{aff}}_e}(a) = \setb{w^{\mathrm{aff}}\in \widehat{W}^{\mathrm{aff}}_e}{w^{\mathrm{aff}}a=a}
\end{equation}
is trivial, $\mathrm{Stab}_{\widehat{W}^{\mathrm{aff}}_e}(a)=1$. If $a\in F^{\vee} \setminus \int(F^{\vee}) $ then it holds
 \begin{equation}\label{stabEd2}
 \abs{\mathrm{Stab}_{\widehat{W}^{\mathrm{aff}}}(a)} = 2\abs{\mathrm{Stab}_{\widehat{W}^{\mathrm{aff}}_e}(a)},
\end{equation}
where
\begin{equation}\label{stabd}
 \mathrm{Stab}_{\widehat{W}^{\mathrm{aff}}}(a) = \setb{w^{\mathrm{aff}}\in \widehat{W}^{\mathrm{aff}}}{w^{\mathrm{aff}}a=a}.
\end{equation}
\end{enumerate}
\end{tvr}

\section{Grids $F^{e}_M$ and $\Lambda^{e}_M$ }

\subsection{Grid $F^{e}_M$}\

The grid $F^{e}_M$ is the finite fragment of the lattice $\frac{1}{M}P^{\vee}$ which is found inside of $F^e$. Suppose we have a fixed $M\in \N$ and consider the $W$-invariant group $\frac{1}{M}P^{\vee}/Q^{\vee}$. The group $\frac{1}{M}P^{\vee}/Q^{\vee}$ is finite with the order
\begin{equation}\label{grupa}
\abs{\frac{1}{M}P^{\vee}/Q^{\vee}}=cM^n.
\end{equation}

We define the grid $F^{e}_M$ as such cosets from $\frac{1}{M}P^{\vee}/Q^{\vee}$ which have a representative element in the fundamental domain $F^{e}$:
\begin{equation*}
F^{e}_M\equiv\frac{1}{M}P^{\vee}/Q^{\vee}\cap F^{e}.
 \end{equation*}

From the relation (\ref{rfun1E}), we have that
\begin{equation}\label{WFM}
 W^eF^e_M = \frac{1}{M}P^{\vee}/Q^{\vee}.
 \end{equation}

The grid $F^e_M$ can be viewed as a union of two disjoint grids -- the grid $F_M\equiv \frac{1}{M}P^{\vee}/Q^{\vee}\cap F$ and the reflection $r_j$ of its interior $\wt{F}_M\equiv \frac{1}{M}P^{\vee}/Q^{\vee}\cap \int(F)$,
 \begin{equation}\label{disFM}
F^{e}_M=F_M\cup r_j\wt{F}_M.
 \end{equation}

We obtain from (\ref{deffun}) that the set $F_M$, or more precisely its representative points, can be identified as
\begin{equation}\label{FM}
 F_M = \setb{\frac{s_1}{M}\om^{\vee}_1+\dots+\frac{s_n}{M}\om^{\vee}_n}{s_0,s_1,\dots ,s_n \in \Z^{\geq 0},\, s_0+\sum_{i=1}^n s_im_i=M}.
 \end{equation}
The reflection $r_j$ of its interior $\wt{F}_M$ is given by
\begin{equation}\label{riFM}
 r_j \wt F_M = \setb{\frac{s'_1}{M}\om^{\vee}_1+\dots+ \frac{s'_j}{M}(\om^{\vee}_j-\al^{\vee}_j)+\dots+\frac{s'_n}{M}\om^{\vee}_n}{s'_0,s'_1,\dots ,s'_n \in \N,\, s'_0+\sum_{i=1}^n s'_im_i=M}.
 \end{equation}

\subsection{Number of elements of $F^{e}_M$}\

The number of elements of $F^e_M$ could be obtained by combining results from \cite{HP}.
\begin{tvr}\label{numE}
Let $m$ be the Coxeter number. Then
$$| F^e_M|=\begin{cases}\abs{F_M} & M<m \\ \abs{F_M}+1 & M=m \\ \abs{F_M}+\abs{F_{M-m}} & M>m.\end{cases}$$
\end{tvr}
\begin{proof}
Considering the equality $|\wt F_M|=|r_j \wt F_M|$, we obtain from the disjoint decomposition (\ref{disFM})
that
\begin{equation}\label{soucet}
\abs{F^{e}_M}=\abs{F_M}+ |\wt{F}_M|.
\end{equation}
It was shown in Proposition 3.5 in \cite{HP}, that
\begin{equation}\label{intcases}
|\wt F_M|=\begin{cases}0 & M<m \\ 1 & M=m \\ |F_{M-m}| & M>m .\end{cases}.
\end{equation}
\end{proof}

\begin{thm}\label{numAn}
The numbers of points of the grid $F^e_M$ of Lie algebras $A_n$, $B_n$, $C_n$, $D_n$ are given by the following relations:
\begin{enumerate}
 \item $A_n,\,n\geq 1$ $$|F^e_M(A_n)|=\begin{pmatrix}n+M\\ n \end{pmatrix}+\begin{pmatrix}M-1\\ n \end{pmatrix}$$
\item  $C_n,\,n\geq 2$
\begin{eqnarray*}
|F^e_{2k}(C_n)|=& \begin{pmatrix}n+k \\ n \end{pmatrix}+\begin{pmatrix}n+k-1 \\ n \end{pmatrix}+\begin{pmatrix}k \\ n \end{pmatrix}+\begin{pmatrix}k-1 \\ n \end{pmatrix}\\
|F^e_{2k+1}(C_n)|=& 2\begin{pmatrix}n+k \\ n \end{pmatrix}+2\begin{pmatrix}k \\ n \end{pmatrix} \end{eqnarray*}
\item $B_n,\,n\geq 3$ $$|F^e_M(B_n)|=|F^e_M(C_n)|$$
\item $D_n,\,n\geq 4$\begin{eqnarray*}
|F^e_{2k}(D_n)|=& \begin{pmatrix}n+k \\ n \end{pmatrix}+6\begin{pmatrix}n+k-1 \\ n \end{pmatrix}+\begin{pmatrix}n+k-2 \\ n \end{pmatrix}+\begin{pmatrix}k+1 \\ n \end{pmatrix}+6\begin{pmatrix}k \\ n \end{pmatrix}+\begin{pmatrix}k-1 \\ n \end{pmatrix}\\
|F^e_{2k+1}(D_n)|=& 4\begin{pmatrix}n+k \\ n \end{pmatrix}+4\begin{pmatrix}n+k-1 \\ n \end{pmatrix}+4\begin{pmatrix}k+1 \\ n \end{pmatrix}+4\begin{pmatrix}k \\ n \end{pmatrix} \end{eqnarray*}
\end{enumerate}
\end{thm}
\begin{proof}
For the case $A_n$, we have from \cite{HP} that $m=n+1$ and $|F_M(A_n)|=\begin{pmatrix}n+M\\ n \end{pmatrix}$. It can be verified directly that the formula
 $$|\wt F_M(A_n)|=\begin{pmatrix}M-1\\ n \end{pmatrix}$$
satisfies (\ref{intcases}) for all values of $M\in \N$. The result follows from (\ref{soucet}). Analogously we obtain formulas for algebras $B_n$, $C_n$ and $D_n$.
\end{proof}

Using explicit formulas for $\abs{F_M}$ from \cite{HP}, the number $|F^e_M|$ of the five exceptional Lie algebras can be obtained similarly from Proposition \ref{numE}.

\begin{example}\label{ex1}
For the Lie algebra $C_2$, we have Coxeter number $m=4$ and $c=2$. Consider for example $M=4$. For the order of the group $\frac{1}{4}P^{\vee}/Q^{\vee}$, we have from (\ref{grupa}) that $\abs{\frac{1}{4}P^{\vee}/Q^{\vee}}=32$ and according to Theorem \ref{numAn} we calculate $\abs{F^e_4(C_2)}=
10. $
The coset representants of $\frac{1}{4}P^{\vee}/Q^{\vee}$ and the fundamental domain $F^e$ are depicted in Figure \ref{figC2}.
\begin{figure}
\resizebox{11cm}{!}{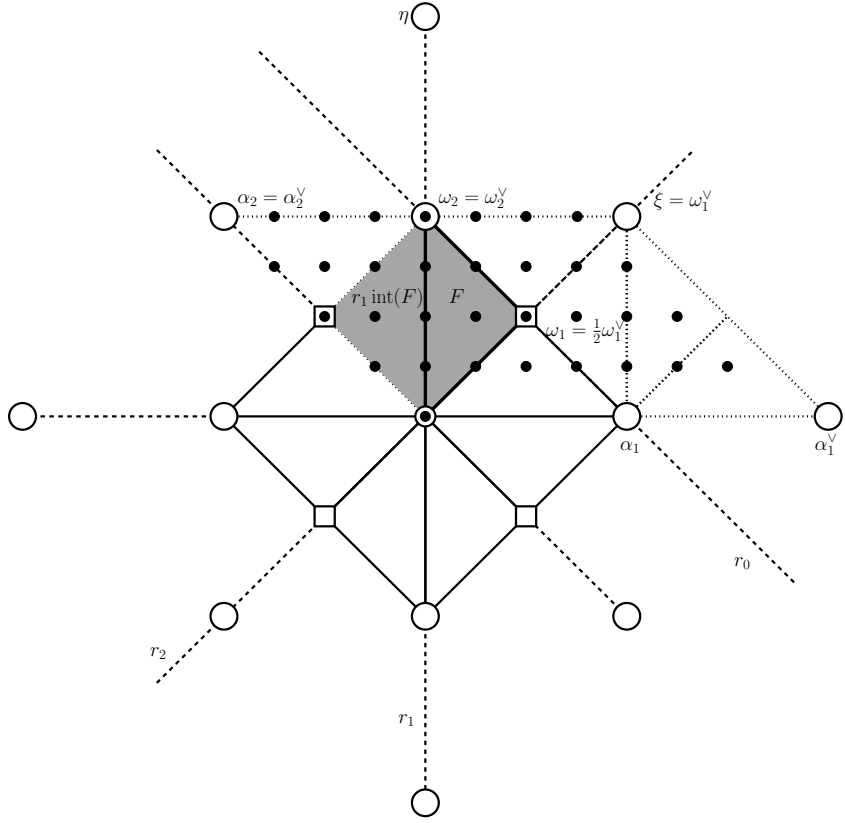}
\caption{ Coset representants of $\frac{1}{4}P^{\vee}/Q^{\vee}$ of $C_2$; coset representants are shown as $32$ black dots, the grey area is the fundamental domain $F^e=F\cup r_1 \int (F)$ containing $10$ points of $F^e_4(C_2)$. Dashed lines represent 'mirrors' $r_0,r_1$ and $r_2$. Circles are elements of the root lattice $Q$, together with squares they are elements of the weight lattice $P$.  }\label{figC2}
\end{figure}
\end{example}

\subsection{Grid $\Lambda^e_M$}\

The points of $\Lambda^e_M$ are the weights that specify $E-$functions belonging to the same pairwise orthogonal set. Further on, we consider $E-$functions that are sampled on the points $F^e_M$. We consider the lowest possible set of such points. The number of points of $\Lambda^e_M$ coincides with the number of points of $F^e_M$.

The $W$-invariant group $P/MQ$ is isomorphic to the group $\frac{1}{M}P^{\vee}/Q^{\vee}$ and its order is given by
\begin{equation*}
\abs{P/MQ}=cM^n.
\end{equation*}
Define the grid $\Lambda^e_M$ as such cosets from $P/MQ$ with representative elements in~$M F^{e\vee}$
\begin{equation*}
 \Lambda^e_M\equiv M F^{e\vee} \cap P/MQ.
 \end{equation*}

The grid $\Lambda^e_M$ can be viewed as a union of two disjoint grids -- the grid $\Lambda_M\equiv P/MQ\cap M F^{\vee}$ and the reflection $r_j$ of its interior $\wt{\Lambda}_M\equiv P/MQ\cap \int(M F^{\vee})$,
 \begin{equation}\label{disLAM}
\Lambda^{e}_M=\Lambda_M\cup r_j\wt{\Lambda}_M.
 \end{equation}

We have from (\ref{defdfun}) that the set $\Lambda_M$, or more precisely its representative points, can be identified as
\begin{equation}\label{LAM}
\Lambda_M = \setb{t_1\om_1+\dots+t_n\om_n}{t_0,t_1,\dots ,t_n \in \Z^{\geq 0},\, t_0+\sum_{i=1}^n t_im^{\vee}_i=M}.
 \end{equation}
The reflection $r_j$ of its interior is given by
\begin{equation}\label{riLAM}
r_j\wt\Lambda_M = \setb{t'_1\om_1+\dots+t'_j(\om_j-\al_j) +\dots +t'_n\om_n}{t'_0,t'_1,\dots ,t'_n \in \N,\, t'_0+\sum_{i=1}^n t'_im^{\vee}_i=M}.
 \end{equation}
Since the $n$--tuple of dual marks $(m^{\vee}_1,\dots,m^{\vee}_n)$ is a certain permutation of the $n$--tuple $(m_1,\dots, m_n)$, we have from (\ref{FM}) and (\ref{LAM}) that $|F_M|=|\Lambda_M|$, and from (\ref{riFM}) and (\ref{riLAM}) that $\abs{r_j\wt\Lambda_M}=\abs{r_j\wt F_M}$. Taking into account disjoint decompositions (\ref{disFM}) and (\ref{disLAM}), we conclude that
\begin{equation}\label{FLE}
 |F^e_M|=|\Lambda^e_M|.
\end{equation}

\subsection{Action of $W^e$ on $P/MQ$}\

If we have two elements $b,b'\in \R^n$ such that $b'-b=Mq$, with $q\in Q$ then for $w\in W^e$ we have $wb-wb'=wMq\in MQ$, i.e. we have a natural action of $W^e$ on the quotient group $\R^n/MQ$. For $\la \in \R^n/MQ$ we denote the order of the stabilizer
\begin{equation}\label{hla}
h^{e\vee}_{\la}\equiv |\mathrm{Stab}^{\vee}_e (\la)|,\q \mathrm{Stab}^{\vee}_e (\la)=\set{w\in W^e}{w\la=\la}.
\end{equation}
\begin{tvr}
\begin{enumerate}
\item For any $\la\in P/MQ$ there exists $\la'\in\Lambda^e_M  $ and $w\in W^e$ such that
\begin{equation}\label{dfun1}
 \la=w\la'.
\end{equation}
\item If $\la,\la'\in \Lambda^e_M $ and $\la'=w\la$, $w\in W^e$ then
\begin{equation}\label{lrfun2}
 \la'=\la=w\la.
\end{equation}

\item If $\la\in M F^{e\vee} \cap \R^n/MQ $, i.e. $\la=b+MQ$, $b\in MF^{e\vee}$ then
\begin{equation}\label{rfunstab2}
\mathrm{Stab}^{\vee}_e (\la) \cong \mathrm{Stab}_{\widehat{W}^{\mathrm{aff}}_e}(b/M).
\end{equation}
\end{enumerate}
\end{tvr}
\begin{proof}
\begin{enumerate}
 \item Let $\la\in P/MQ$ be of the form $\la=p+MQ $, $p\in P$. From (\ref{fun1Ed}) it follows that there exist $p'\in F^{e\vee}$, $w\in W^e$ and $q\in Q$ such that
\begin{equation*}
 \frac1M p=w p'+q,
 \end{equation*}
i.e. $ p=wMp'+Mq$.
From $W$--invariance of $P$ we have that $Mp'\in P$, the class $\la'=Mp'+MQ$ is from $\Lambda^e_M$
and (\ref{dfun1}) holds.
\item Let $\la,\la'\in P/MQ$ be of the form $ \la=p+MQ$, $\la'=p'+MQ$ and $p,p'\in MF^{e\vee}$. Suppose that
\begin{equation*}\label{ddfun4}
 p'=wp+Mq,\q q\in Q,\, w\in W^e.
\end{equation*}
Then $p/M,p'/M\in F^{e\vee}$ and it follows from (\ref{fun2Ed}) that $p=p'$.
\item We have from (\ref{directEd}) that for any $w^{\mathrm{aff}}\in \widehat{W}^{\mathrm{aff}}_e$ there exist unique $w\in W^e$ and unique shift $T(q)$ such that $w^{\mathrm{aff}}=T(q)w$. Define a homomorphism $\psi: \mathrm{Stab}_{\widehat{W}^{\mathrm{aff}}_e}(b/M)\map W^e $ for $w^{\mathrm{aff}}\in \mathrm{Stab}_{\widehat{W}^{\mathrm{aff}}_e}(b/M)$ by $\psi(w^{\mathrm{aff}})=\psi(T(q)w)=w$.
If $b/M=w^{\mathrm{aff}}(b/M)=w(b/M)+q$ then $b-wb=Mq\in MQ$, i.e. we obtain $w\in \mathrm{Stab}^{\vee}_e (\la)$ and vice versa. Thus, $\psi (\mathrm{Stab}_{\widehat{W}^{\mathrm{aff}}_e}(b/M))=\mathrm{Stab}^{\vee}_e(\la)$ holds. We also have
\begin{equation*}\label{kerd}
\mathrm{ker}\,\psi= \left\{ T(q)\in \mathrm{Stab}_{\widehat{W}^{\mathrm{aff}}_e}(b/M) \right\}=1.
\end{equation*}
\end{enumerate}
\end{proof}


\begin{example}\label{exdual}
For the Lie algebra $C_2$ we have  $\abs{P/4Q}=32$ and according to (\ref{FLE}) we have $$\abs{\Lambda^e_4(C_2)}=\abs{F^e_4(C_2)}=10. $$
The coset representants of $P/4Q$, the dual fundamental domain $F^{e\vee}$ and the grid $\Lambda^e_4(C_2)=4F^{e\vee}\cap P/4Q$ are depicted in Figure~\ref{figC2d}.
\begin{figure}
\resizebox{11cm}{!}{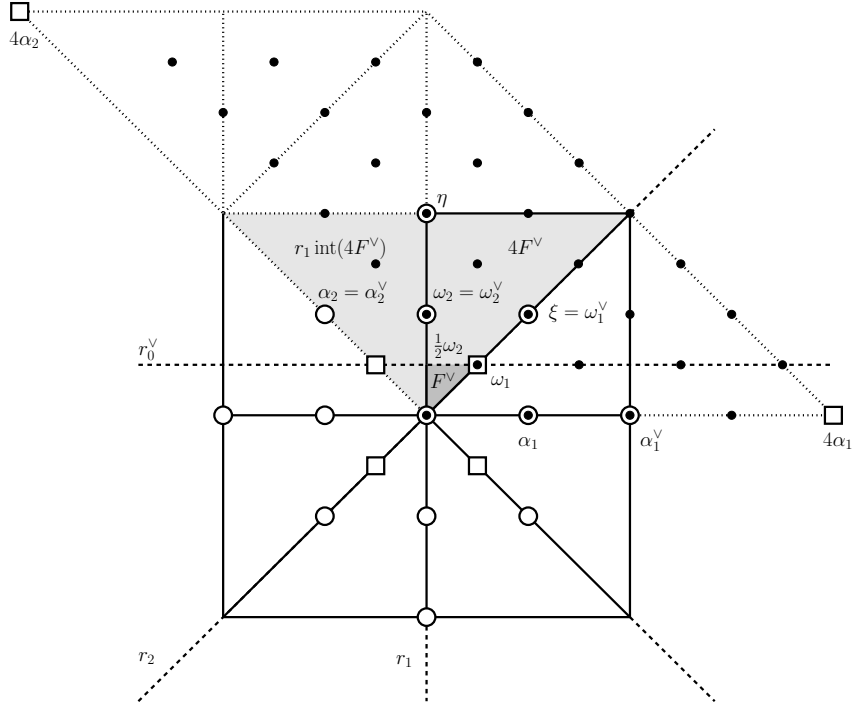}
\caption{ The cosets representants of $P/4Q$ of $C_2$; the cosets representants are shown as $32$ black dots, the darker grey area is the fundamental domain $F^{\vee}$, the lighter grey area is the domain $4F^{e\vee}=4F^{\vee}\cup r_1 \int(4F^{\vee})$ which contains $10$ elements of $\Lambda^e_4(C_2)$. The dashed lines represent dual 'mirrors' $r^\vee_0,r_1$, $r_2$. The circles and squares coincide with those in Figure \ref{figC2}. }\label{figC2d}
\end{figure}
\end{example}

\subsection{Calculation of $h^e_x$ and $h^{e\vee}_\la$}\

Calculation procedure of $h_x\equiv \abs{\mathrm{Stab}_{W^{\mathrm{aff}}}(x)}$ for any $x\in F_M$ and of $h^{\vee}_\la\equiv\abs{ \mathrm{Stab}_{\widehat{W}^{\mathrm{aff}}}(\la)}$, $\la\in \Lambda_M$ was derived in \S 3.7 of \cite{HP}. These calculation procedures use extended Coxeter-Dynkin diagrams ($\mathrm{DD}$) of $G$ and their dual versions $\mathrm{DD}^{\vee}$ (see e.g.\cite{HP}). Modifying these procedures by using the relations (\ref{rfunstabE}), (\ref{stabE2}) and (\ref{rfunstab2}), (\ref{stabEd2}), we deduce a calculation procedure for $h^e_x$, $h^{e\vee}_\la$, defined by (\ref{hxe}), (\ref{hla}), for $x\in F^e_M$ and $\la\in \Lambda^e_M$.

Consider a point $x\in F^e_M=F_M \cup r_j\wt F_M$.
\begin{enumerate}
 \item If $x\in r_j\wt F_M$ then $h^e_x=1$.
 \item Let $[s_0,\dots,s_n]$ be corresponding coordinates of $x\in F_M $ from (\ref{FM}). If $s_0,\dots,s_n$ are all non-zero, then $h^e_x=1$.
 \item If some of the coordinates $[s_0,\dots,s_n]$ are zero then consider such a subgraph $U$ of extended $\mathrm{DD}$ consisting only of those nodes $i$ for which $s_i=0,\, i=0,\dots,n$. The subgraph $U$ consists in general of several connected components $U_l$. Each component~$U_l$ is a (non-extended) $\mathrm{DD}$ of some compact simple Lie group~$G_l$. Take corresponding orders of the Weyl groups $|W_l|$ of $G_l$. Then it holds $$h^e_x=\frac{1}{2}\prod_{l}|W_l|.$$
 \end{enumerate}

 We proceed similarly to determine $h^{e\vee}_\la$ when considering a point $\la\in \Lambda^e_M=\Lambda_M \cup r_j\wt \Lambda_M$.
\begin{enumerate}
 \item If $\la\in r_j\wt\Lambda_M$ then $h^{e\vee}_\la=1$.
 \item Let $[t_0,\dots,t_n]$ be the corresponding coordinates of $\la\in \Lambda_M $ from (\ref{LAM}). If $t_0,\dots,t_n$ are all non-zero then $h^{e\vee}_\la=1$.
 \item If some of the coordinates $[t_0,\dots,t_n]$ are zero then consider such a subgraph $U'$ of the extended $\mathrm{DD}^{\vee}$ consisting only of those nodes $i$ for which $t_i=0,\, i=0,\dots,n$. The subgraph $U'$ consists in general of several connected components $U'_l$. Each component~$U'_l$ is a (non-extended) $\mathrm{DD}$ of some compact simple Lie group~$G'_l$. Take corresponding orders of the Weyl groups $|W'_l|$ of $G'_l$. Then it holds $$h^{e\vee}_\la=\frac{1}{2}\prod_{l}|W'_l|.$$
\end{enumerate}

\section{$W^e-$Invariant functions}
The numbers $h^e_x$, $h^{e\vee}_\la$ and $|F^e_M|$, which were determined so far, are important for the properties of special functions, called $E-$functions when they are sampled on $F^e_M$. A detailed review of the properties of $E-$functions may be found in \cite{KP3}. In this section the goal is to complete and make explicit the orthogonality properties of $E-$functions \cite{MP2}.

\subsection{$E-$functions}\

We recall the definition of $E-$functions and show that they can be labeled by the finite set $\Lambda^e_M$ when sampled on the grid $F^e_M$.

Consider $ b\in P$ and recall that (normalized) $E-$functions can be defined as a mapping $\Xi_b:\R^n\map \Com$
\begin{equation}\label{E}
 \Xi_b(a)=\sum_{w\in W^e} e^{2 \pi \ii \sca{ wb}{a}}.
 \end{equation}
The following properties of $E-$functions are crucial
\begin{itemize}
\item symmetry with respect to $w\in W^e$
\begin{align}\label{Esym}
 \Xi_b(wa)&=\Xi_b(a) \\
\Xi_{wb}(a)&=\Xi_b(a) \label{Esymla}
 \end{align}
\item invariance with respect to $q^{\vee} \in Q^{\vee}$
\begin{equation}\label{Einv}
 \Xi_b(a+q^{\vee})=\Xi_b(a).
 \end{equation}
\end{itemize}

We investigate values of $E-$functions on the grid $F^e_M$. Suppose we have a fixed $M\in \N$ and $s\in \frac{1}{M}P^{\vee}$. From (\ref{Einv}) it follows that we can consider $\Xi_b$ as a function on classes $\frac{1}{M}P^{\vee}/Q^{\vee}$. From (\ref{rfun1E}) and (\ref{Esym}) it follows that we can consider $\Xi_b$ only on the set $F^e_M\equiv\frac{1}{M}P^{\vee}/Q^{\vee}\cap F^e$. We also have
\begin{equation*}
 \Xi_{b+MQ}(s)=\Xi_b(s),\ s\in F^e_M
 \end{equation*}
and thus we can consider the functions $\Xi_\la$ on $F^e_M$ parameterized by classes from $\la\in P/MQ$. Moreover, from (\ref{dfun1}) and (\ref{Esymla}) it follows that {\it we can consider $E-$functions $\Xi_\la$ on $F^e_M$ parameterized by $\la\in\Lambda^e_M$ only}.

\section{Discrete orthogonality of $E-$functions}

\subsection{Basic discrete orthogonality relations}\

Discrete orthogonality of $E-$functions was discussed in general in \cite{MP2}. Practical application of \cite{MP2} is not completely straightforward. Therefore we reformulate the basic facts and subsequently use them to make the discrete orthogonality over $F^e_M$ described in all detail. Basic orthogonality relations from \cite{MP2,HP} are for any $\la,\la' \in P/MQ$ of the form:
\begin{equation}\label{bdis}
 \sum_{y\in\frac{1}{M}P^{\vee}/Q^{\vee}} e^{2\pi\ii\sca{\la-\la'}{y}}=cM^n\delta_{\la,\la'}.
\end{equation}

\subsection{Discrete orthogonality of $E-$functions}\

We define the scalar product of two functions $f,g:F^e_M\map \Com$ by
\begin{equation}
 \sca{f}{g}_{F^e_M}= \sum_{x\in F^e_M}\ep^e(x) f(x)\overline{g(x)},
\end{equation}
where the numbers $\ep^e (x)$ are determined by (\ref{epE}). We show that $\Lambda^e_M$, defined by (\ref{disLAM}), is the lowest maximal set of pairwise orthogonal $E-$functions.
\begin{tvr}
For $\la,\la' \in\Lambda^e_M$ it holds
\begin{equation}\label{orthoE}
 \sca{\Xi_\la}{\Xi_{\la'}}_{F^e_M}=c\abs{W^e}M^n h^{e\vee}_\la \delta_{\la,\la'},
\end{equation}
where $c$, $h^{e\vee}_\la$, $\Xi_\la$ were defined by (\ref{Center}), (\ref{hla}), (\ref{E}) respectively, $n$ is the rank of $G$.
\end{tvr}
\begin{proof}
The equality
\begin{equation*}
\sum_{x\in F^e_M} \ep^e(x) \Xi_\la(x)\overline{\Xi_{\la'}(x)}= \sum_{y \in \frac{1}{M}P^{\vee}/Q^{\vee}} \Xi_\la(y)\overline{\Xi_{\la'}(y)}
\end{equation*}
follows from (\ref{rfun2E}) and (\ref{WFM}) and $W^e-$invariance of the expression $\Xi_\la(x)\overline{\Xi_{\la'}(x)}$. Then, using $W^e-$invariance of $\frac{1}{M}P^{\vee}/Q^{\vee}$ and (\ref{bdis}), we have
\begin{align*}
\sca{\Xi_\la}{\Xi_{\la'}}_{F^e_M}=&\sum_{w'\in W^e}\sum_{w\in W^e} \sum_{y\in \frac{1}{M}P^{\vee}/Q^{\vee}}e^{2\pi\ii\sca{w\la-w'\la'}{y}}=\abs{W^e}\sum_{w'\in W^e}\sum_{y \in \frac{1}{M}P^{\vee}/Q^{\vee}}e^{2\pi\ii\sca{\la-w'\la'}{y}}\\ =& c\abs{W^e}M^n \sum_{w'\in W^e} \delta_{w'\la',\la}.
\end{align*}
Since $\la,\la' \in\Lambda^e_M$ we have from (\ref{lrfun2}) that $$ \sum_{w'\in W^e} \delta_{w'\la',\la}=h^{e\vee}_{\la}\delta_{\la,\la'}. $$
\end{proof}

\begin{example}
The highest root $\xi$ and the highest dual root $\eta$ of $C_2$ are determined by the formulas
$$\xi=2\al_1+\al_2,\ \ \eta=\al^\vee_1+2\al^\vee_2.$$
The even Weyl group of $C_2$ has four elements, $|W^e|=4$, and we obtain for the determinant of the Cartan matrix 
$c=2$.
Thus, decomposition of the grid $F^e_M(C_2)=F_M(C_2)\cup r_1 \wt{F}_M(C_2)$ is given by
\begin{align} F_M(C_2) =& \setb{\frac{s_1}{M}\om^{\vee}_1+\frac{s_2}{M}\om^{\vee}_2}{s_0,s_1,s_2\in \Z^{\geq 0},\, s_0+2s_1+s_2=M}\label{C2Me} \\
 r_1\wt F_M(C_2) =& \setb{\frac{-s'_1}{M}\om^{\vee}_1+\frac{s'_2+2s'_1}{M}\om^{\vee}_2}{s'_0,s'_1,s'_2\in \N,\, s'_0+2s'_1+s'_2=M} \nonumber
\end{align}

and the grid of weights $\Lambda^e_M(C_2)=\Lambda_M(C_2)\cup r_1 \wt{\Lambda}_M(C_2)$ is determined by
\begin{align}
\Lambda_M(C_2) =& \setb{t_1\om_1+t_2\om_2}{t_0,t_1,t_2\in \Z^{\geq 0},\, t_0+t_1+2t_2=M}\label{C2Le}\\
 r_1\wt\Lambda_M(C_2) =& \setb{-t'_1\om_1+(t'_1+t'_2)\om_2}{t'_0,t'_1,t'_2\in \N,\, t'_0+t'_1+2t'_2=M}. \nonumber
\end{align}

The discrete orthogonality relations of $E-$functions of $C_2$ which hold for any two functions $\Xi_\la$, $\Xi_{\la'}$  labeled by $\la,\la'\in\Lambda^e_M(C_2)$ are of the form (\ref{orthoE}). The coefficients $\ep^e (x)$, $h_\lambda ^{e\vee}$, which appear in (\ref{orthoE}), have according to \S 3.5 values $\ep^e (x)=4$, $h_\la ^{e\vee}=1$ for $x\in r_1\wt F_M(C_2)$, $\la\in r_1\wt{\Lambda}_M(C_2)$. We represent each point $x\in F_M(C_2)$ and each weight $\lambda\in \Lambda_M(C_2)$ by the coordinates $[s_0,s_1,s_2]$ and $[t_0,t_1,t_2]$ from relations (\ref{C2Me}) and (\ref{C2Le}), respectively. The values of the coefficients $\ep^e (x)$, $h_\lambda ^{e\vee}$ for $x\in F_M(C_2) $, $\la\in \Lambda_M(C_2)$ are listed in Table \ref{F4}.
\begin{table}[h]
\begin{tabular}{c|c}
$x\in F_M(C_2)$ & $\ep^e (x)$  \\ \hline
$[s_0,s_1,s_2]$ & $4$  \\
$[0,s_1,s_2]$ & $4$   \\
$[s_0,0,s_2]$ & $4$   \\
$[s_0,s_1,0]$ & $4$   \\
$[0,0,s_2]$ & $1$   \\
$[0,s_1,0]$ & $2$  \\
$[s_0,0.0]$ & $1$  \\

\end{tabular}\hspace{24pt}
\begin{tabular}{c|c}
$\lambda\in \Lambda_M(C_2)$  & $h_\lambda ^{e\vee}$ \\ \hline
$[t_0,t_1,t_2]$ & $1$  \\
$[0,t_1,t_2]$ & $1$   \\
$[t_0,0,t_2]$ & $1$   \\
$[t_0,t_1,0]$ & $1$   \\
$[0,0,t_2]$ & $2$   \\
$[0,t_1,0]$ & $4$  \\
$[t_0,0,0]$ & $4$  \\
\end{tabular}
\caption{The coefficients $\ep^e (x)$ and $h_\lambda ^{e\vee}$ of $C_2$.  Assuming $s_0,s_1,s_2 >0 $, $t_0,t_1,t_2 >0 $. }\label{F4}
\end{table}
\end{example}


\subsection{Discrete $E-$transforms}\

Analogously to ordinary Fourier analysis, we define interpolating functions $ \Xi^M$
\begin{align}
\Xi^M(x):=& \sum_{\la\in \Lambda^e_M} c_\la \Xi_\la(x),\q x\in \R^n \label{intc}
\end{align}
which are given in terms of expansion functions $\Xi_\la$ and expansion coefficients $c_\la$, whose values need to be determined. These interpolating functions can also be understood as finite cut-offs of infinite expansions.

Next we discretize the equations (\ref{intc}). Suppose we have some function $f$ sampled on the grid $F^e_M$. The interpolation of $f$ consists in finding the coefficients $c_\la$ in the interpolating functions (\ref{intc}) such that
\begin{align}
\Xi^M(x)=& f(x), \q x\in F^e_M \label{intcs}
\end{align}
Relations (\ref{FLE}) and (\ref{orthoE}) allow the values $\Xi_\la(x)$ with $x\in F^e_M$, $\la\in \Lambda^e_M$ to be viewed as elements of a non-singular square matrix. This invertible matrix coincides with the matrix of the linear system (\ref{intcs}). Thus, the coefficients $c_\la$ can be uniquely determined. The formula for calculation of $c_\la$, which is also called discrete $E-$transform, can be obtained by means of calculation of standard Fourier coefficients:
\begin{align}
c_\la=& \frac{\sca{f}{\Xi_\la}_{F^e_M}}{\sca{\Xi_\la}{\Xi_\la}_{F^e_M}}=(c\abs{W^e} M^nh^{e\vee}_\la)^{-1}\sum_{x\in F^e_M}\ep^e(x) f(x)\overline{\Xi_\la(x)}\label{ctrans}
\end{align}
We also have the corresponding Plancherel formula
\begin{equation*}
\sum_{x\in F^e_M} \ep^e(x)\abs{f(x)}^2=c \abs{W^e} M^n \sum_{\la\in\Lambda^e_M}h^{e\vee}_\la\abs{c_\la}^2 .
\end{equation*}

\section{Concluding Remarks}
The $E-$functions of the article have other undoubtedly useful properties that were not invoked here. Let us briefly mention six:
\begin{itemize}

\smallskip
\item
Product of two $E-$functions with the same underlying Lie group and the same arguments $x\in\R^n$ but different subscripts, say $\lambda$ and $\lambda'$, decompose into the sum of $E-$functions. This is a powerful property that enables building of a set of recursion relations for constructing even larger $E-$functions.

\smallskip\item
Points of the weight lattice $P$ of $G$ split into a few disjoint congruence classes. It is convenient to specify a congruence class of $\mu\in P$ by its congruence number \cite{LP}.  There is a different linear function of $\mu$ for different Lie group $G$. The weights of one $W$-orbit, hence also the weights of one $W^e-$orbit, belong to the same congruence class. As a consequence, all summands $e^{2\pi\ii\langle\mu,x\rangle}$ in an $E-$function have the weights $\mu$ from the same congruence class. Hence, the $E-$function has a well defined congruence number. During multiplication of $E-$functions of the same $G$, the congruence numbers of the $E-$functions add up. In particular, the decomposition of a product of two $E-$functions must contain summands with the same congruence number. The number of congruence classes of a $G$ is equal to the order of the center of the compact simple Lie group $G$.

\smallskip\item
Another undoubtedly useful property of orbit functions, which plays no role in this paper, is the fact that they are eigenfunctions of the Laplace operator appropriate for the Lie group, and that the eigenvalues are known in all cases \cite{KP3}.

\smallskip\item
The $E-$functions considered in the paper so far have the underlying simple group $G$. Suppose now that it is semisimple but not simple, say $G=G_1\times G_2$, where $G_1$ and $G_2$ are simple. Then there are two options as to what to take for the even subgroup $W^e(G_1\times G_2)$. The simpler of the two is to define $W^e(G_1\times G_2):=W^e(G_1)\times W^e(G_2)$. The $E-$functions of $W^e(G_1)\times W^e(G_2)$ are then products of the $E-$functions of $G_1$ and $G_2$, etc., see  \cite{KP3}. Such an option is trivial, for example when $G=SU(2)\times SU(2)$.

The option where $W^e(G_1\times G_2)$ is bigger, namely the full even subgroup of $W(G_1)\times W(G_2)$, is somewhat more interesting. This options is still to be explored in the literature \cite{HKP}. It is already non-trivial in the lowest case $G=SU(2)\times SU(2)$.

\smallskip\item
The results presented in this paper are valid for any compact simple Lie group $G$, including the five exceptional cases. Explicit counting formulas for $|F^e_M|$ for these five cases, which were not provided here, can be straightforwardly put together using Proposition \ref{numE} and the Appendix in \cite{HP}.

\smallskip\item
The present work raises the question: under which conditions converge series of the functions $\{\Xi^M\}_{M=1}^\infty$ assigned to a function $f:F^e\map\Com$ by the relations (\ref{intc}), (\ref{ctrans}).

\end{itemize}

\section*{Acknowledgments}
Work was supported  by the Natural Sciences and Engineering Research Council of Canada and in part also by the MIND Research Institute, Santa Ana, California. JH~is grateful for the postdoctoral fellowship and for the hospitality extended to him at the Centre de recherches math\'ematiques, Universit\'e de Montr\'eal.

\end{document}